%% file: main.tex
\documentclass{article}
\usepackage{geometry}[margin=1in]
\usepackage{mdframed}
\usepackage{authblk}
\usepackage{xurl}
\input{settings.tex}

\title{The Monotone Priority System: Foundations of Contract-Specific Sequencing}

\author{Naveen Durvasula}
\affil{Columbia University \& Ritual}
\date{January 2026}

\begin{document}

\maketitle

\begin{abstract}
    Modern blockchain applications benefit from the ability to specify sequencing constraints on the transactions that interact with them. This paper proposes a principled and axiomatically justified way of adding sequencing constraints on smart contract function calls that balances expressivity with the tractability of block production. Specifically, we propose a system in which contract developers are allowed to set an integer global priority for each of their calls, so long as that the call's chosen priority is no higher than the priority of any of its referenced calls. Block builders must then simply sequence transactions in priority order (from high to low priority), breaking ties however they would like. We show that this system is the unique system that satisfies five independent axioms.
\end{abstract}

\section{Introduction}
\subsection{Contract-Specific Sequencing}

A blockchain network maintains a public and tamper-resistant state that is used to keep track of information with restricted mutability, such as account balances. The network's state is updated by transactions that are submitted by users. A transaction may be a payment, or more generally, any state transition that is compliant with the blockchain's execution semantics. As modern blockchain protocols support Turing-complete execution, a transaction may make a complex sequence of function calls as it is executed. Transactions are batched into blocks due to the blockchain's underlying consensus algorithm. Each block is an ordered sequence of transactions that additionally satisfies some validity constraints that are set by the protocol. Standard block validity constraints include limits on the aggregate network resources utilized by the transactions within the block, and requirements that constituent transactions be well-formed \cite{wood2014ethereum}.

While blocks must adhere to these validity constraints, there is typically a single block producer who retains broad freedom to select which transactions are included in the block, and the order in which they are executed. The profit that the block producer can obtain by making use of this freedom is known as MEV (Maximum Extractable Value), and is the subject of a rich line of literature spanning multiple disciplines -- see e.g. \cite{heimbach2023ethereum, bahrani2024transaction, malkhi2022maximal, kelkar2020order} and references therein. It is now understood that there may be value in allowing an application to directly impose additional constraints on the ordering of the transactions within a block that interact with that application's smart contract over the course of its execution. Such constraints can allow applications to improve their user experience by mitigating the negative effects induced by adversarial transactions that are submitted by other users (or a profit-maximizing block producer) \cite{sorellalabs2024neweraass, eclipselabs2024deepdiveass}. In particular it's emerged as a crucial feature for decentralized exchanges that clear using order books\footnote{For the curious reader, we give an extended explanation of how additional sequencing constraints can improve user experience in a decentralized order book. Suppose there is a price shock that causes the price of the asset being traded on the order book to drop by a large amount. Standing buy orders now remain in the book at a price that is higher than the users who submitted them want to pay. This induces a race between those with outstanding orders who would like to cancel them, versus those who want to fill those orders at a stale higher price. The outstanding orders typically belong to market-makers who provide liquidity for the exchange by serving as the counter-party to the typical user's trades. It's therefore widely believed that exchanges can be more successful by prioritizing cancellations, as providing a better experience for market-makers is considered more important than allowing an arbitrageur to fill a stale order. Such prioritization has similarly been proposed for traditional (centralized) order books \cite{baldauf2020high}. Contract-specific sequencing would allow the decentralized order book to directly impose the constraint that cancellations are prioritized over fill orders within a block, as opposed to allowing a profit-maximizing block producer to ultimately decide the victor. Indeed, some of the most popular decentralized exchanges already do this in ad-hoc ways \cite{hyperliquid2024latency_tx_ordering, cavey2025ace_amqs}.} \cite{hyperliquid2024latency_tx_ordering, cavey2025ace_amqs, yakovenko2025internet_capital_markets_roadmap}. The goal of this paper is to provide a principled and axiomatically justified approach to implementing a system that allows smart contract developers to impose additional sequencing constraints. 

There are a variety of different approaches one could take in allowing contracts to impose sequencing constraints. On one extreme, we could allow contracts to specify arbitrary sequencing constraints over any two transactions. While this gives contract developers complete control over the sequencing process, it makes the process of block building wholly computationally intractable. Indeed, finding a valid block that includes a collection of transactions amounts to solving a general Boolean satisfiability problem, and a valid block need not even exist! On the other extreme is the status quo, where block building remains tractable but contract developers have no ability to directly sequence transactions. The key challenge in designing a contract-specific sequencing system is balancing the tension between the expressivity of the system (i.e. how rich are the contract developers' choices over sequencing constraints) and tractability (i.e. how easy is it for block producers to build a block that satisfies these additional constraints). Even supposing that a contract can only impose constraints on transactions that directly interact with it, it is not immediately obvious how to build a tractable system when a single transaction may make nested sub-calls to numerous contracts over the course of its execution. Should all of those contracts have the right to impose constraints on the sequencing of that transaction? And if so, how can we ensure that block building remains tractable in light of that? An ideal sequencing system must make restrictions to remain tractable while still satisfying some basic desiderata of allowing a contract to impose constraints on the sequencing of transactions that interact with it. In this paper, we propose a system that we believe lies in a sweet spot, and in some sense is maximally expressive under a certain definition of sequencing being tractable. We call this system the Monotone Priority System.

\subsection{Monotone Priorities: Simple, but Expressive Sequencing Constraints} \label{sec:intro-mps}
The Monotone Priority System operates within a setting in which there are smart contracts that each have some constituent function calls, or executable blocks of code that can modify the state of the blockchain network. Crucially, function calls are composable, in that they can call and execute other function calls over the course of their execution. As hinted at in the previous subsection, composability is what makes it challenging to balance expressivity and tractability: because a single call can touch many contracts through nested sub-calls, it is tempting to design the system to be expressive enough to allow all of those contracts to impose some constraints on the sequencing of a transaction that makes that call. However, building a block that satisfies all of those constraints concurrently must remain simple. 

In the Monotone Priority System, we allow each contract to independently select an integer global \textit{priority} for each of its constituent calls, where a transaction that makes a higher priority root call should be sequenced before a transaction that makes a lower priority root call. The contract developer can select any priority for each call, subject to just one simple constraint:
\[\text{a call cannot have higher priority than any of the calls it references}\]
A block is considered to be valid if it orders transactions in order from high priority to low priority, breaking ties arbitrarily. To prevent contract developers from assigning arbitrarily high priorities to a function call, it may make sense in practice to implement a maximum priority $\lambda_{\texttt{max}}$ that no contract developer can choose to exceed.

To check that a call's priority is valid, we only need to check that its priority is not higher than any of the calls it directly makes. This recursively guarantees that a call will not have have higher priority than any nested sub-calls that it makes, no matter how deeply nested those sub-calls are. A contract developer is allowed to set any priority for a call so long as it is valid.

\begin{figure}[h!]
\centering
\begin{mdframed}[innerleftmargin=30pt, innerrightmargin=30pt]
\begin{center}
    \textbf{The Monotone Priority System: Valid Priorities}
\end{center}
    
    \vspace{15pt}
    A call with a priority $\lambda$ is $\texttt{VALID}$ if:
    \begin{enumerate}
    \item the call only references $\texttt{VALID}$ calls
    \item no call it references has priority greater than $\lambda$
    \end{enumerate}
\end{mdframed}
 
\end{figure}
At a high level, the Monotone Priority System imposes this constraint in order to prevent users from circumventing a contract's priorities through composability. More concretely, suppose that a contract $x$ wants to prioritize one of its calls, $a$, over another one of its calls, $b$, and it does so by making $a$ have a higher priority than $b$. If we don't require that calls have $\texttt{VALID}$ priorities, then a user who wants to make a call to $b$ without their transaction being sequenced after calls to $a$ can do so by making a call to $b'$ in a new contract $x'$. The call $b'$ just executes the behavior of $b$ by calling it directly. If $x'$ can set any priority for $b'$, then it can be prioritized above the call $a$ from contract $x$. The $\texttt{VALID}$ constraint prevents this type of behavior. 

\subsubsection{Expressivity}\label{sec:intro-mps-exp}
How expressive is this system? Even though contracts can't arbitrarily choose the priority of each call (since it must be $\texttt{VALID}$), they can still freely relatively order the priorities of their constituent calls. To make this more concrete, suppose as before that there is a contract\footnote{An example of such a contract is the decentralized order book from the previous footnote, where $a$ is the function call for placing a cancellation, and $b$ is the function call for placing a fill order.} $x$ with two function calls $a$ and $b$, and the contract would like $a$ to have higher priority than $b$. If neither $a$ nor $b$ make any calls to other contracts, it is straightforward to see that $x$ can prioritize $a$ over $b$. This is because it is $\texttt{VALID}$ for $x$ to choose any priorities it would like for the two calls (e.g. $0$ for $a$ and $-1$ for $b$). However, this remains true even if $a$ and $b$ make arbitrary calls to other contracts. Suppose for example that $b$ still does not call any other contracts, but $a$ calls a single helper function\footnote{In the running order book example, a cancellation call might call an escrow function in another contract to release tokens that were locked in an outstanding buy order.} $c$ in another contract $x'$, and the contract $x'$ set the priority of $c$ to be $-100$. While $x$ must set the priority of $a$ to be no higher than $-100$ in order for it to be $\texttt{VALID}$, $x$ can still make $b$ have a lower priority than $a$, for example, by setting it to $-101$. More generally, a contract can always set the relative ordering of its calls' priorities however it would like while keeping them $\texttt{VALID}$, because the Monotone Priority System only prevents contracts from setting priorities that are too \textit{high}. A contract can always arbitrarily deprioritize one call relative to another in order to achieve a desired relative ordering.

Once a contract has chosen a relative ordering over its calls, the Monotone Priority System allows that contract to make the following two guarantees, regardless of the choices made by other contracts. Specifically, for any call $a$ that the contract would like to prioritize over a call $b$, the contract can guarantee that:
\begin{enumerate}
    \item \textbf{Transactions making root calls to $a$ are sequenced before transactions making root calls to $b$ in any valid block.} This is because the contract can always set $\texttt{VALID}$ priorities for $a$ and $b$ so that $b$ has a lower priority than $a$, as we discussed above. 
    
    \item \textbf{Transactions making root calls to $a$ are sequenced before any transactions that potentially call $b$ at some point in their full execution trace.} Any call $c$ that potentially references $b$ during its execution must have a priority that is no higher than $b$ in order to be $\texttt{VALID}$. Since $a$ has a higher priority than $b$, it also will have a higher priority than $c$.
\end{enumerate}

The expressivity of the system can be thought of in terms of these two guarantees. The second guarantee captures how the Monotone Priority System handles composability: every contract that a complex call interacts with can indeed impose some sequencing constraints on that call in a restricted way. Those restrictions keep the system simple, as we discuss next. 

\subsubsection{Simplicity}\label{sec:intro-mps-sim}
The Monotone Priority System converts the sequencing constraints of each of the contracts into a single block validity constraint that remains tractable to satisfy. From the perspective a block builder who has selected a given collection of transactions to include, a valid block can be found by just sorting the transactions in descending order of priority, breaking ties in any way desired. The simplicity of the aggregate ordering constraint makes composition with existing block validity constraints, such as constraints on the cumulative ``size'' of (i.e. maximum network resources used by) the transactions within the block, straightforward.  We exhibit a simple algorithm that finds a valid sequence of transactions under the Monotone Priority System\footnote{Enumerating over all tie-breaking rules $\tau$ allows one to find all valid blocks that include all of the transactions.}.

\begin{figure}[ht!]
\centering
\begin{mdframed}[innerleftmargin=30pt, innerrightmargin=30pt]
\begin{center}
    \textbf{The Monotone Priority System: Valid Blocks}
\end{center}
\vspace{15pt}
\textbf{Input:} A collection of transactions $t_1, \dots, t_n$ that each make a root call to a contract with $\texttt{VALID}$ priorities.

\vspace{5pt}
\textbf{Output:} A transaction ordering that is valid under the Monotone Priority System.

\vspace{7pt}
\hrule

\vspace{7pt}
Initialize $\tau: [n] \longleftrightarrow [n]$ to be a tie-breaking rule. Then,

\begin{enumerate}
    \item Sort each $t_i$ first in descending order by the priority $\lambda_i$ of its root call.
    \item Secondarily sort transactions $t_i$ with the same priority by $\tau(i)$
\end{enumerate}

\textbf{Return:} The resulting ordered sequence of transactions 
\end{mdframed}
\end{figure}

The Monotone Priority System also makes it simple for contract developers to reason about how modifications to their code affect the sequencing of transactions that call the developer's contract. Sequencing constraints depend only on the name of a call, and not on the current network state. Furthermore, it is easy to reason about how priorities change under composition. The contract developer faces no additional sequencing constraints on its calls if it chooses to include a new reference to a high priority call. However, if the contract developer includes a reference to a low priority call, then the developer's call must also have low priority. Put differently, the maximum possible priority of a call can only get lower as that call includes more references to other calls. 

\subsection{Alternate Designs}

Is it possible to make a system that is more expressive without making block building intractable? The main technical result that we prove in a formal sense in this paper is no. More specifically, we define five natural axioms, and show that the Monotone Priority System is the unique system that satisfies all of the axioms. We next give informal descriptions of all five axioms:

\paragraph{Existence.} Our first axiom, existence, requires that no matter what sequencing constraints each contract places on transactions, there must exist a block that orders all of the transactions in a way that satisfies all of the constraints. Without existence, contracts may be able to select conflicting sequencing constraints (e.g. if one contract requires that a transaction $t$ comes before $t'$, and another requires that a transaction $t'$ comes before $t$), resulting in no valid block that both includes all of the transactions and satisfies all of the sequencing constraints. 

\paragraph{Priority.} Our second axiom, priority, requires that the system must allow each contract to relatively order any transactions making root calls to any of the contract's constituent calls, regardless of the constraints that have been imposed by other contracts. This gives the system a baseline level of expressivity.

\paragraph{Extension.} Our third axiom, extension, requires that if the sequencing system allows a contract to sequence a call $c$ before $c'$, then it must also extend that constraint to sequence $c$ before any other call that references $c'$. Without extension, it would be very easy to get around the system's sequencing constraints: a user can just make a new dummy contract that calls $c'$ to get the same functionality as $c'$, but without the constraint that the call must be sequenced after $c$.

\paragraph{Reducibility.} Our fourth axiom, reducibility, states that a contract should not be allowed to sequence a transaction $t'$ to be after another transaction $t$ unless $t'$ references one of the contract's calls in its full trace that the contract also sequences after $t$. Without reducibility, a contract developer may have a difficult time understanding what rights other contracts can impose over the sequencing of a transaction that calls the developer's contract. 

\paragraph{Independence of Irrelevant Calls.} Our fifth and final axiom, independence of irrelevant calls, states that whether the system grants a contract the option to sequence one transaction over another (perhaps, contingent on other contract's choices) must depend only on the calls that are referenced in the full traces of those two transactions. Without it, a contract developer would have to reason about irrelevant calls in other contracts to determine whether they can impose a sequencing constraint between two particular transactions. \\

The remainder of the paper is organized as follows. In Section \ref{sec:model}, we formally define the full design space of contract-specific sequencing rights systems, as well as the five axioms we've listed above. This design space is rich, where the protocol can allow contracts to make any collection of sequencing constraints in a way that can be contingent on the choices made by other contracts. In Section \ref{sec:mps}, we formally describe the expressivity guarantees from Section \ref{sec:intro-mps-exp} within this language. We then show that the system implementing those guarantees satisfies all five axioms in Section \ref{sec:mps-sat} (Theorem \ref{thm:sat}), and is further unique in doing so in Section \ref{sec:mps-uni} (Theorem \ref{thm:uni}). In Section \ref{sec:ind}, we show that all five axioms are in fact necessary for uniqueness: there exist other systems that satisfy any collection of four, but not the fifth (Theorem \ref{thm:ind}). Next, in Section \ref{sec:order}, we discuss when a system of sequencing constraints can be implemented by assigning global integer priorities to each function call. We show that this is not guaranteed by the Existence axiom (Theorem \ref{thm:order}), but it is implied by all five axioms together (Theorem \ref{thm:mpsexist}). We then formally show that the Monotone Priority System implements the expressivity guarantees from Section \ref{sec:intro-mps-exp} in Section \ref{sec:order-mps} (Theorem \ref{thm:mpsimp}). Finally, in Section \ref{sec:dis}, we give a brief discussion and open directions. 

\section{Model} \label{sec:model}

\subsection{Calls, Contracts, and References}
We consider a setting in which there is a set of \textit{calls} $\C$ and a set of \textit{contracts} $\mb X$. Associated with any call $c \in \C$ is a set of calls $r(c) \subseteq \C$ that the call directly references (i.e. makes a root call to). Given the reference map, we can define the \textit{call trace} of a call $\tr: \C \to 2^\C$ as follows:

\begin{equation}
    \tr(c) := \set{c' \in \C \mid \exists (c_0 := c, c_1, c_2, \dots, c_N) \text{ s.t. } c_N = c' \text{ and }c_{i} \in r(c_{i-1})}
\end{equation}

The call trace includes all calls that may be referenced during the complete execution of a call. 
\subsection{Rights Systems and Axioms}\label{sec:model-axioms}

A \textit{rights system} $\R(p)$ is a nonempty family of sets parameterized a map $p: \C \to \mb X$, consisting of vectors of the form $(\succ_x)_{x \in \mb X}$. The map $p$ indicates the parent contract for each call. Each $\succ_x$ is a strict partial order on the set $\C$ of all calls, and can be interpreted as a sequencing constraint chosen by contract $x$. We make an additional constraint on the parent mappings $p$ that parameterize a rights system. We require the existence of a \textit{deployment order} $>_D$: a strict total ordering over the set of contracts $\mb X$ satisfying the property
\begin{equation}
    x >_D x' \implies c \not \in \tr(c') \text{ for all }c, c' \in \C \text{ s.t. }p(c) = x\text{ and }p(c') = x'
\end{equation}
In other words, a contract cannot have a child call that references a child call of a contract that is deployed in the future. Next, we require that for any contract $x$, if there is a contract $x' >_D x$ that is deployed after $x$, then that contract only has a finite number of child calls $\abs{\C_x(p)} < \infty$. We call such a $p$ \textit{deployable}.

We now list five axioms that one may wish such a system to abide by. Our first axiom is \textit{existence}, which asserts that for any finite collection of transactions (each transaction is represented as a call), and any collection of strict partial orders that were selected from the rights system, there exists a block ordering $B$ that sequences all of the transactions without violating the sequencing constraints that were chosen by any contract. 

\begin{axiom}[Existence]
For all deployable $p: \C \to \mb X$, $(\succ_x)_{x \in \mb X} \in \R(p)$ and all finite sequences $(t_1, \dots, t_N)$ where each $t_i \in \C$, there exists as bijection $B: [N] \longleftrightarrow [N]$ such that for all $i, j$ where $i < j$, there does not exist an $x \in \mb X$ such that $t_{B(j)} \succ_x t_{B(i)}$.
\end{axiom}

Our second axiom is \textit{priority}, which states that each contract has the power to sequence transactions that directly call its child calls, regardless of the sequencing constraints chosen by any of the other contracts. In order to formally define it, we make a few intermediate definitions. Given a pointer map $p: \C \to \mb X$, we define for each $x \in \mb X$ as shorthand
\[\C_x(p) := \set{c \in \C \mid p(c) = x}\]
That is, $\C_x(p)$ are the calls that are directly included in the contract $x$ under a pointer map $p$. We call a strict partial order $\succ^*_x$ over $\C_x(p)$ \textit{admissible} if there does not exist a $c, c' \in \C_x(p)$ such that $c \succ^*_x c'$ and $c \in \tr(c)$. In other words, an admissible partial order over a contract $x$'s child calls under $p$ does not allow $x$ to prioritize a call $c$ over $c'$ if $c$ calls $c'$. The priority axiom allows each contract to sequence its child calls according to any admissible partial order, regardless of the constraints chosen by other contracts:

\begin{axiom}[Priority]
For all deployable $p: \C \to \mb X$, $(\succ_x)_{x \in \mb X} \in \R(p)$, all contracts $x^* \in \mb X$, and all admissible strict partial orders $\succ^*_{x^*}$ over the set $\C_x(p)$ of child calls of contract $x^*$, there exists a $(\succ'_x)_{x \in \mb X} \in \R(p)$ such that for all $x \ne x^*$, $\succ'_x = \succ_x$, and $c_1 \succ'_{x^*} c_2 \iff c_1 \succ^*_{x^*} c_2$ for all $c_1, c_2 \in \C_x(p)$.
\end{axiom}

Our third axiom is \textit{extension}, which states that if the contract $x^*$ chooses to sequence a call $c$ to have higher priority than a call $c'$ under the rights system, then the sequencing constraint must extend to any other call $c''$ which calls $c'$ at some point in its trace. Without the extension axiom, a transaction may avoid the sequencing constraint by using a dummy call in a different contract.

\begin{axiom}[Extension]
For all deployable $p: \C \to \mb X$, $(\succ_x)_{x \in \mb X} \in \R(p)$, all contracts $x^* \in \mb X$, and all calls $c, c' \in \C$, $c \succ_{x^*} c' \implies c \succ_{x^*} c''$, for any $c''$ such that $c' \in \tr(c'')$.
\end{axiom}

Our fourth axiom is \textit{reducibility}, which states that if a contract $x$ would like to sequence a call $c$ before $c'$, then there must be some sub-call $c''$ that is made during the execution of $c'$ that $x$ must also have the right to sequence before $c$. Reducibility is a property that may be crucial to developers building such a $c'$: they may guarantee with it that $x$ does not have the right to sequence a call $c$ before it by ensuring that $c'$ does not reference any of $x$'s child calls that are sequenced after $c$. 

\begin{axiom}[Reducibility]
    For all deployable $p: \C \to \mb X$, $(\succ_x)_{x \in \mb X} \in \R(p)$, all contracts $x^* \in \mb X$, and all calls $c, c' \in \C$, if $c \succ_{x^*} c'$, then either $p(c') = x^*$ or there exists a $c'' \in \tr(c')$ such that $p(c'') = x^*$ and $c \succ_{x^*} c''$.
\end{axiom}

Our fifth and final axiom is \textit{independence of irrelevant calls}. It states that if a contract $x$ can sequence a call $t$ prior to $t'$, then (perhaps contingent on the choices of other contracts) it must be theoretically possible for $x$ to sequence $t$ prior to $t'$ if the pointer map changes only for calls that are irrelevant to $t$ and $t'$.

\begin{axiom}[Independence of Irrelevant Calls]
Let $p: \C \to \mb X$ be deployable, $(\succ_x)_{x \in \mb X} \in \R(p)$, $x^* \in \mb X$, and $t, t' \in \C$ be such that $t \succ_{x^*} t'$. Then, for any deployable $p': \C \to \mb X$ such that $p'(t) = p(t)$, $p'(t') = p(t')$, and $p(c) = p'(c)$ for all $c \in \tr(t) \cup \tr(t')$, there exists a $(\succ'_x)_{x \in \mb X} \in \R(p')$ such that $t \succ'_{x^*} t'$.
\end{axiom}

\section{A Unique Rights System} \label{sec:mps}

We define a rights system $\mb R^*(p)$ to be in bijection with the set of all $(\succ^*_x)_{x \in \mb X}$ such that each $\succ^*_x$ is an admissible strict partial order on $\C_x(p)$. For any such $(\succ^*_x)_{x \in \mb X}$, we let there be a $(\succ_x)_{x \in \mb X} \in \R^*(p)$ where $t \succ_x t'$ if and only if
\begin{enumerate} 
\item $p(t) = p(t') = x$ and $t \succ^*_x t'$, or
\item $p(t) = x$ and $p(t') \ne x$, and there exists a $c \in \tr(t')$ with $p(c) = x$ such that $t \succ^*_x c$
\end{enumerate}

Notice that these are the same two guarantees as we informally stated in Section \ref{sec:intro-mps-exp}. We show in this section that $\R^*(p)$ is unique among rights systems that satisfy the axioms in Section \ref{sec:model-axioms}. 

\subsection{Satisfaction} \label{sec:mps-sat}

Next, we show that $\R^*(p)$ as described satisfies all five axioms. Our proof that $\R^*(p)$ satisfies Existence (Lemma \ref{lem:exist}) works by utilizing the correspondence to the Monotone Priority System that we establish in the previous subsection. 

\begin{theorem}\label{thm:sat}
    For any $p: \C \to \mb X$ with a deployment order $>_D$, $\R^*(p)$ satisfies Existence, Priority, Extension, Reducibility, and Independence of Irrelevant Calls.
\end{theorem}

\begin{proof}
    \begin{lemma}\label{lem:exist}
        $\R^*(p)$ satisfies Existence.
    \end{lemma}
    \begin{proof}
        This follows directly from Theorems \ref{thm:mpsexist} and \ref{thm:order}, which are discussed in Section \ref{sec:order}. 
    \end{proof}

    \begin{lemma}\label{lem:priority}
        $\R^*(p)$ satisfies Priority.
    \end{lemma}
    \begin{proof}
        This follows trivially from the fact that every $(\succ^*_x)_{x \in \mb X}$ where $\succ_x^*$ is an admissible strict partial order over $\C_x(p)$ has an extension in $\R^*(p)$. 
    \end{proof}

    \begin{lemma}\label{lem:ext}
        $\R^*(p)$ satisfies Extension.
    \end{lemma}
    \begin{proof}
        Let $p: \C \to \mb X$, and $(\succ_x)_{x \in \mb X} \in \R(p)$. Next, let $x^* \in \mb X$ and $t, t' \in \C$ such that $t \succ_x t'$. Let $t'' \in \C$ be a call such that $t' \in \tr(t'')$. It then follows that there exists a $c \in \tr(t'')$ such that $p(c) = x^*$, whence we must have that $t \succ_x t''$ as desired.
    \end{proof}

    \begin{lemma}\label{lem:red}
    $\R^*(p)$ satisfies Reducibility.
    \end{lemma}
    \begin{proof}
    Reducibility follows trivially from the definition of $\R^*(p)$: For all deployable $p: \C \to \mb X$, $(\succ_x)_{x \in \mb X} \in \R(p)$, all contracts $x^* \in \mb X$, and all calls $t, t' \in \C$, if $t \succ_{x^*} t'$, then either $p(t') = x^*$ or there exists a $c \in \tr(t')$ such that $p(c) = x^*$ and $t \succ_{x^*} c$.
    \end{proof}

    \begin{lemma}\label{lem:iic}
    $\R^*(p)$ satisfies Independence of Irrelevant Calls.
    \end{lemma}
    \begin{proof}
        Let $p: \C \to \mb X$, $(\succ_x)_{x \in \mb X} \in \R(p)$, $x^* \in \mb X$, and $t, t' \in \C$ be such that $t \succ_{x^*} t'$. By the definition of $\R^*(p)$, there are then two exhaustive cases.
        \begin{case}
            $p(t) = p(t') = x^*$
        \end{case}
        In this case, take any $p': \C \to \mb X$ such that $p'(t) = p(t)$, $p'(t') = p(t')$, and $p(c) = p'(c)$ for all $c \in \tr(t) \cup \tr(t')$. As $p'(t) = p'(t') = x$, and $\succ_{x^*}$ was admissible over $\C_{x^*}(p)$, there exists a $(\succ'_x)_{x \in \mb X} \in \R(p')$ such that $t \succ'_{x^*} t'$.
        \begin{case}
            $p(t) = x^*$, $p(t') \ne x^*$, and there exists a $c \in \tr(t')$ with $p(c) = x^*$ such that $t \succ_{x^*} c$.
        \end{case}
        In this case, take any $p': \C \to \mb X$ such that $p'(t) = p(t)$, $p'(t') = p(t')$, and $p(c) = p'(c)$ for all $c \in \tr(t) \cup \tr(t')$. For such a $p'$, we must still have that $p'(t) = p'(c) = x^*$. As $\succ_{x^*}$ was admissible over $\C_{x^*}(p)$, there exists a $(\succ'_x)_{x \in \mb X} \in \R(p')$ such that $t \succ'_{x^*} c$. It follows that $t \succ'_{x^*} t'$, as $c \in \tr(t')$ remains true.  
    \end{proof}
\end{proof}

\subsection{Uniqueness}\label{sec:mps-uni}

Next, we show that $\R(p)$ is the unique rights system that satisfies all five axioms. In order to prove this result, we need to make an assumption of richness on the space of calls, namely that one can always make a new call that is distinct from a collection of existing calls that references yet another existing call. 

\begin{assumption}[Richness]\label{ass:rich}
For any finite collection of calls $N \subseteq \C$ and a call $c \in \C$, there exists a call $c' \in \C \setminus N$ such that $c \in r(c')$.
\end{assumption}

Turing-complete execution environments, in particular, satisfy this basic richness assumption. 

\begin{theorem}[Uniqueness]\label{thm:uni}
    Let $\R(p)$ be any rights system satisfying Existence, Priority, Extension, Reducibility, and Independence of Irrelevant Calls. Then, for any $p: \C \to \mb X$, $\R(p) = \R^*(p)$.
\end{theorem}
\begin{proof}
    Suppose first for contradiction that for some $p: \C \to \mb X$, there is a rights system $\R(p)$ satisfying all five axioms such that $\R(p) \not \subseteq \R^*(p)$. Then, there exists a $(\succ_x)_{x \in \mb X} \in \R(p)$ such that $(\succ_x)_{x \in \mb X} \not \in \R^*(p)$. 

    \begin{lemma}\label{lem:admissible}
        For each $x \in \mb X$, $\succ_x$ is an admissible strict partial order when restricted to $\C_x(p)$.
    \end{lemma}

    \begin{proof}
        Suppose for contradiction that this is false. Then, there exists an $x \in \mb X$ and $c, c' \in \C_x(p)$ such that $c \succ_x c'$ and $c' \in \tr(c)$. As $\R(p)$ satisfies Extension, it follows that $c \succ_x c$, which is a contradiction.
    \end{proof}

    It follows by Lemma \ref{lem:admissible} that there exists a $(\succ^*_x)_{x \in \mb X}$ such that each $\succ_x^*$ is an admissible strict partial order on $\C_x(p)$, and for any $x \in \mb X$ and $c, c' \in \C_x(p)$, $c \succ_x^* c' \iff c \succ_x c'$. We now let $(\succ'_x)_{x \in \mb X} \in \R^*(p)$ be the unique vector of strict partial orders that also agrees with $(\succ^*_x)_{x \in \mb X}$ on each $\C_x(p)$. To achieve contradiction, we show that for any $x \in \mb X$ and $t, t' \in \C$, $t \succ_x t' \iff t \succ'_x t'$.

    There are three relevant cases to consider. 

    \begin{case}
        $p(t) = p(t') = x$
    \end{case}

    It trivially holds that $t \succ_x t' \iff t \succ'_x t'$, as both $\succ_x$ and $\succ'_x$ agree with $\succ^*_x$ for any such $t, t' \in \C_x(p)$

    \begin{case}
        $p(t) = x$ and $p(t') \ne x$. 
    \end{case}

    We first show that $t \succ_x t' \implies t \succ'_x t'$. Suppose for contradiction that this is not true. Then, $t \succ_x t'$ and $t \not \succ'_x t'$. It follows that there does not exist a $c \in \tr(t)$ with $p(c) = x$ such that $t \succ^*_x c$. However, this would imply that $\R(p)$ does not satisfy reducibility, leading to a contradiction. 

    Next, we show that $t \succ'_x t \implies t \succ_x t$. Suppose for contradiction that this is not true. Then, $t \succ'_x t'$ and $t \not \succ_x t'$. It follows that there exists a $c \in \tr(t')$ with $p(c) = x$ such that $t \succ^*_x c$. It further follows by the definition of $(\succ^*_x)_{x \in \mb X}$ that $t \succ_x c$. However, now we find that $\R(p)$ does not satisfy Extension, leading to a contradiction.

    \begin{case}
        $p(t) \ne x$
    \end{case}

    As $t \not \succ'_x t'$ for all $t, t'$ with $p(t) \ne x$, it suffices to show that $t \not \succ_x t'$. To see that this is true, suppose for contradiction that $t \succ_x t'$. By Assumption \ref{ass:rich}, there exists a call $c' \in \C$ such that $t' \in r(c')$ and $c' \not \in \tr(t) \cup \tr(t') \cup \set{t,t'}$. Next, since $\R(p)$ satisfies Independence of Irrelevant Calls, there exists a $(\comp \succ_x)_{x \in \mb X} \in \R(p')$ such that $t \comp \succ_x t'$, where
    \[p'(c) := \begin{cases} p(t) & c = c'\\ p(c) & \text{else}\end{cases}\]
    Let $x' := p(t)$. By Priority, there then exists a $(\comp \succ_x')_{x \in \mb X} \in \R(p')$ such that $t \comp \succ'_x t'$ and $t' \comp \succ'_{x'} t$. Next, consider the sequence of transactions $(t_1, t_2) := (t, t')$. For any $B: [2] \longleftrightarrow [2]$, it is therefore true that $t_{B(2)} \comp \succ'_x t_{B(1)}$ or $t_{B(2)} \comp \succ'_{x'} t_{B(1)}$. However, $\R(p)$ must now therefore not satisfy Existence, yielding a contradiction as desired. 

    With the conclusion of this case, we have shown that $\R(p) \subseteq \R^*(p)$. Next, we show that $\R^*(p) \subseteq \R(p)$. To see this, suppose for contradiction that there exists a $(\succ_x)_{x \in \mb X} \in \R^*(p)$ that is not in $\R(p)$. It follows that for some $(\succ^*_x)_{x \in \mb X}$ where each $\succ^*_x$ is an admissible strict partial order over $\C_x(p)$, there does not exist $(\succ'_x)_{x \in \mb X} \in \R(p)$ that agrees with $(\succ^*_x)_{x \in \mb X}$ on each $\C_x(p)$. However, as $\R(p)$ must be non-empty, we may repeatedly apply Priority to find such a $(\succ'_x)_{x \in \mb X} \in \R(p)$, yielding a contradiction. 
\end{proof}

\section{Independence}\label{sec:ind}

In this section, we show that our five axioms are independent of each other.

\begin{theorem}\label{thm:ind}
Existence, Priority, Extension, Reducibility, and Independence of Irrelevant Calls are independent axioms.
\end{theorem}
\begin{proof}
    To prove this statement, we need to exhibit a model that satisfies any four of the axioms without satisfying the fifth. As such, this proof is split into five lemmas. 

    \begin{lemma}
        There exists a model that does not satisfy Existence, but satisfies Priority, Extension, Reducibility, and Independence of Irrelevant Calls.
    \end{lemma}

    \begin{proof}
        Let $\C := \set{a, a', b}$, where $r(a) = r(b) = \emptyset$, and $r(a') = \set{a}$. Let $\mb X  = \set{A, B}$. Given any deployable $p: \C \to \mb X$, we define a rights system $\R(p)$ as follows. Contract $A$ can select any admissible strict partial order $\succ_A$ over $\C_A(p)$. Then, for any $c \in \C_A(p)$ and $c' \in \C_B(p)$, we let $c \succ_A c'$ if and only if there exists a $c'' \in \C_A(p)$ such that $c' \in \tr(c'')$. Contract $B$ similarly selects a $\succ_B$. However, $B$ can additionally (optionally) choose to add any constraints of the form $c \succ_B c'$, where $p(c') = B$, and $c \not \in \tr(c')$. For any such constraints that are added, the constraints $c \succ_b c''$ must also be added, for any $c''$ such that $c' \in \tr(c'')$. If the corresponding $\succ_B$ remains a strict partial order, then $B$ may choose it. 

        The rights system $\R^*(p)$ trivially satisfies Priority (as $B$ can choose to take no additional constraints). It further satisfies Extension, as we have explicitly closed all $\succ_A$ and $\succ_B$ under this constraint. Next, notice that $A$'s choices of $\succ_A$ trivially satisfy reducibility. $B$'s choices also satisfy reducibility, as if it adds any additional constraints $c \succ_B c'$, then either $p(c') = B$ or there exists some $c''$ such that $c \succ_B c''$, $p(c'') = B$, and $c'' \in \tr(c')$. Finally, the rights system satisfies Independence of Irrelevant Calls, as whether or not any contract $x \in \mb X$ can choose $c \succ_x c'$ for some $c, c' \in \C$ depends only on the value of $p$ on the set $\set{c, c'} \cup \tr(c) \cup \tr(c')$.

        However, consider the particular parent mapping
        \[p^*(c) := \begin{cases}A & c \in \set{a', b}\\B & c = a \end{cases}\]
        This mapping is deployable with the ordering $A >_D B$. Under the rights system $\R^*(p^*)$, however, it is possible to choose $\succ_B$ such that $b \succ_B a$ and $b \succ_B a'$, while also choosing $\succ_A$ such that $a' \succ_A b$. This contradicts Existence, as for the sequence $(t_1, t_2) := (a', b)$ of transactions, for any $B: [2] \longleftarrow [2]$, there will exist some contract $x \in \mb X$ such that $t_{B(2)} \succ_x t_{B(1)}$.
    \end{proof}

    \begin{lemma}
        There exists a model that does not satisfy Priority, but satisfies Existence, Extension, Reducibility, and Independence of Irrelevant Calls.
    \end{lemma}

    \begin{proof}
        Take $\C := \set{a,b,c,d}$ where $r(t) = \emptyset$ for all $t \in \C$, $\mb X := \set{A, B}$, and let $\R(p)$ have only one element $(\emptyset_A, \emptyset_B)$, where each $\emptyset_x$ is the null ordering. This trivially satisfies Existence (as it makes no constraints), Extension (as there are no constraints to extend), Reducibility (as there are no constraints to reduce), and Independence of Irrelevant Calls (as it is constant). It does not satisfy Priority, as for any $(\succ^*_A, \succ^*_B)$ where each $\succ^*_x$ is a strict partial ordering over $\C_x(p)$ (all such orderings are admissible in this model), and some $\succ^*_x$ is not the null ordering, there does not exist a $(\succ_A, \succ_B) \in \R(p)$ where $\succ_x = \succ^*$ when restricted to $\C_x(p)$. 
    \end{proof}

    \begin{lemma}
        There exists a model that does not satisfy Extension, but satisfies Existence, Priority, Reducibility, and Independence of Irrelevant Calls.
    \end{lemma}

    \begin{proof}
        Let $\C := \set{a, a', b}$ and $X := \set{A, B}$ where $r(a) = r(b) = \emptyset$ and $r(a') = a$. We then define $\R(p)$ to be a rights system where we include all $(\succ_A, \succ_B)$ such that $\succ_A$ is exactly a strict partial order on $\C_A(p)$ and $\succ_B$ is exactly a strict partial order on $\C_B(p)$, with no constraints $c \succ_x c'$ allowed where $p(c) \ne p(c')$ for any $x \in \mb X$. 

        This system satisfies Existence, as we can arbitrarily order the contracts, and then order any sequence of transactions $(t_1, \dots ,t_N)$ first by their contract under $p$, and then by the partial order chosen by that contract. The system trivially satisfies Priority and Reducibility (as it only allows $c \succ_x c'$ if $p(c') = x$). Finally, it satisfies Independence of Irrelevant Calls, as a contract $x$'s choice to sequence $c \succ_x c'$ depends only on $p(c)$ and $p(c')$. 

        However, $\R(p)$ does not satisfy Extension. Choosing \[p^*(c) := \begin{cases}A & c = a'\\B & c = \set{a,b} \end{cases}\]
        (which can be easily verified as deployable), while there exists a $\succ_B$ such that $b \succ_B a$, there is no $\succ_B$ such that $b \succ_B a'$ despite the fact that $a \in \tr(a')$.
    \end{proof}

    \begin{lemma}
        There exists a model that does not satisfy Reducibility, but satisfies Existence, Priority, Extension, and Independence of Irrelevant Calls.
    \end{lemma}

    \begin{proof}
        Let $\C := \set{a', a, b}$ and $\mb X := \set{A, B}$. Let $r(a) = r(b) = \emptyset$ and $r(a') = \set{a}$. Given any deployable $p: \C \to \mb X$, we define a rights system $\R(p)$ as follows. Contract $A$ can select any admissible strict partial order $\succ_A$ over $\C_A(p)$. Then, for any $c \in \C_A(p)$ and $c' \in \C_B(p)$, we let $c \succ_A c'$ if and only if there exists a $c'' \in \C_A(p)$ such that $c' \in \tr(c'')$. Contract $B$ similarly selects a $\succ_B$. However, $B$ can additionally (optionally) choose to add any constraints of the form $c \succ_B c'$, where $p(c) = B$, $p(c') \ne B$, $c \not \in \tr(c')$, and there exists some $c'' \in \tr(c)$ such that $p(c'') = B$. For any such constraints that are added, the constraints $c \succ_b c''$ must also be added, for any $c''$ such that $c' \in \tr(c'')$. If the corresponding $\succ_B$ remains a strict partial order, then $B$ may choose it.

        The system satisfies Existence, because there can be at most one constraint of the form $c \succ_x c'$ where $p(c) \ne p(c')$. This follows from the fact that such a constraint can only arise if there exists some $c'' \in \tr(c')$ such that $p(c'') = p(c)$. There is only one pair of calls ($a$ and $a'$) for which this relationship potentially holds true, and thus such a $c$ and $c'$ cannot exist for both contracts at once. As all remaining constraints are consistent with respect to a partial order on $\C_x$, there exists a total ordering that respects all constraints as desired. The system trivially satisfies Priority and Extension by construction. It also satisfies Independence of Irrelevant Calls, since a contract $x$'s ability to add the constraint $c \succ_x c'$ does not change for any $p'$ that is unchanged on $\set{c, c'} \cup \tr(c')$.

        However, $\R(p^*)$ does not satisfy Reducibility for 
        \[p^*(c) := \begin{cases}A & c = a'\\B & c = \set{a,b} \end{cases}\]
        This is because we can let $a \succ_B b \succ_B a'$ under $\R(p^*)$. This contradicts Reducibility, since $b \succ_B a'$, but there does not exist a $c \in \tr(a')$ such that $b \succ_B c$. 
    \end{proof}

    \begin{lemma}
        There exists a model that does not satisfy Independence of Irrelevant Calls, but satisfies Existence, Priority, Extension, and Reducibility.
    \end{lemma}

    \begin{proof}
        Let $\C := \set{a', a, b}$ and $\mb X := \set{A, B}$. Let $r(a) = r(b) = \emptyset$ and $r(a') = \set{a}$. Given any deployable $p: \C \to \mb X$, we define a rights system $\R(p)$ as follows. Contract $A$ can select any admissible strict partial order $\succ_A$ over $\C_A(p)$. Then, for any $c \in \C_A(p)$ and $c' \in \C_B(p)$, we let $c \succ_A c'$ if and only if there exists a $c'' \in \C_A(p)$ such that $c' \in \tr(c'')$. Contract $B$ similarly selects a $\succ_B$. However, $B$ can additionally (optionally) choose to add the constraints $b \succ_B a$ and $b \succ_B a'$ if and only if $p(b) = A$, $p(a) = B$, and $p(a') = B$. 

        The system satisfies Existence, because if $B$ does not choose to add the extra constraint, then there are only sequencing constraints with respect to child calls of the same contract. If $B$ does choose to add the extra constraint, there cannot be a conflicting constraint $a \succ_A b$ or $a' \succ_A b$ because $p(a) = p(a') \ne A$. The system trivially satisfies Priority and Extension by construction. It also satisfies reducibility, as the original constraints satisfy reducibility, and the extra constraints $b \succ_B a$ and $b \succ_B a'$ can only be added if $p(a) = p(a') = B$. 

        However, $\R(p)$ does not satisfy Independence of Irrelevant Calls, as the constraint $b \succ_B a$ is achievable if $p(a') = B$, but not achievable if $p(a') = A$. Since $a' \not \in \set{a, b} \cup \tr(a) \cup \tr(b)$, Independence of Irrelevant Calls is violated.
    \end{proof}
\end{proof}

\section{Global Ordering}\label{sec:order}
In addition to the axioms we list in Section \ref{sec:model-axioms}, an additional property that enables the tractable block-building process described in Section \ref{sec:intro-mps-sim} is that the sequencing constraints induced by the rights system can be represented by a static global priority over all calls. A global priority order can be thought of as embedding the set of calls $\C$ into a totally ordered set $(A, >)$ such that the total ordering respects all of the sequencing constraints set by each of the contracts.

\begin{definition}[Orderability]
    Let $(A, >)$ be a totally ordered set. A rights system $\R(p)$ is $A$-Orderable if for all deployable $p: \C \to \mb X$ and $(\succ_x)_{x \in \mb X} \in \R(p)$, there exists a $\lambda: \C \to A$ such that for any $c, c' \in \C$ and $x \in \mb X$, $c \succ_x c' \implies \lambda(c) > \lambda(c')$
\end{definition}

It is straightforward to see that Orderability implies Existence. If the set of all calls $\C$ can be totally ordered in a way that respects all of the sequencing constraints selected by contracts under the rights system, then for any finite collection of transactions, a valid block can be built by ordering the transactions by the global ordering. It is tempting to believe that the converse might be true as well: if any finite collection of transactions can be ordered, then perhaps there is also a global ordering that holds over all calls. This is true in a sense, but only if we allow the priority order to embed into the rational numbers $\Q$, which is far too rich to be useful for practical implementations. It is impossible to store arbitrary rational numbers without using infinitely many bits. Thus, for an arbitrary rights system satisfying Existence, there may not be a way to statically assign each call to a global integer priority that would not need to be re-computed each time a new contract is deployed. It is yet an even stronger constraint to require that there is a global integer priority ordering that has a fixed maximum priority $\lambda_{\texttt{max}}$. 

\begin{theorem}\label{thm:order}
A rights system $\R(p)$ satisfies Existence if and only if it is $\Q$-Orderable. However, Existence does not imply $\Z$-Orderability, and $\Z$-Orderability does not imply $\Z_{\le \lambda_{\texttt{max}}}$-Orderability for any $\lambda_{\texttt{max}} \in \Z$. 
\end{theorem} 

\begin{proof}
    We first show that $\Q$-Orderability implies Existence. Let $\R(p)$ be a $\Q$-Orderable rights system. For any $(\succ_x)_{x \in \mb X}$, there must exist a $\lambda: \C \to \Q$ such that for any $c, c' \in \C$ and $x \in \mb X$, $c \succ_x c' \implies \lambda(c) > \lambda(c')$. Next, let $(t_1, \dots, t_N)$ be any finite sequence of calls. Select any $B: [N] \longleftarrow [N]$ that sorts the $t_i$ in descending order by $\lambda(t_i)$. It follows that there does not exist an $i < j$ and $x \in \mb X$ such that $t_{B(j)} \succ_x t_{B(i)}$, whence the desired result follows.

    Next, we show that Existence implies $\Q$-Orderability. First, notice that for any deployable $p$ and any $(\succ_x)_{x \ in \mb X} \in \R(p)$, the binary relation $\succ$ given by the closure of $\bigcup_{x \in \mb X} \succ_x$ under transitivity must be a strict partial order. This is because if it were not irreflexive or antisymmetric, then there must exist some $c_1, \dots, c_k \in \C$ and $x_1, \dots, x_k \in \mb X$ such that 
    \[c_1 \succ_{x_1} \cdots c_k \succ_{x_k} c_1\]
    However, then the rights system $\R(p)$ would fail Existence for the sequence of calls $(c_1, \dots, c_k)$. It is well-known that any strict partial order $\succ$ over $\C$ can be extended to a total order $>$ over $\C$. It is further well-known that any total order over a countable set can be embedded into $\Q$. The result then follows.

    The remaining two statements further follow straightforwardly from well-known results that subsets of $\Q$ cannot necessarily be embedded into $\Z$ while preserving order (and similarly, $\Z$ cannot be embedded into $\Z_{\le \lambda_{\texttt{max}}}$ while preserving order). However, to capture the intuition of these arguments, we explicitly give a counterexample that shows why a rights system that satisfies Existence may not be $\Z$-Orderable.

    To see this, let $\C$ be a countably infinite set, and let $\R(p)$ consist of only one element $(\succ_x)_{x \in \mb X}$. This element is defined as follows: we first let $x^* \in \mb X$ be a specific element of $\mb X$. For all $x \ne x^*$, we define $\succ_x := \emptyset$ to be the null ordering over $\C$ (i.e. all elements of $\C$ are considered incomparable to each other). We then define $\succ_{x^*}$ as follows. First, we let $c_0, c_F \in \C$ to be arbitrary calls in $\C$. Then, we define $\mu: \C\setminus\set{c_0, c_F} \longleftrightarrow \Z$ to be any bijection. Such a bijection must exist as $\C\setminus\set{c_0, c_F}$ is countable. We then define, for any $c, c' \in \C$:
    \[c \succ_{x^*} c' \iff \mu(c) > \mu(c') \text{ or } \bra{c = c_0 \text{ and } c' \ne c_0} \text{ or } \bra{c \ne c_F \text{ and } c' = c_F}\]
    To see that $\R(p)$ satisfies existence, notice that for any finite sequence of calls $(t_1, \dots, t_N)$, we can find an appropriate bijection $B$ by first sorting all $t_i \not \in \set{c_0, c_F}$ in descending order by $\mu(t_i)$, and then sequencing all $t_i = c_0$ at the beginning, and all $t_i = c_F$ at the end. However, this $\R(p)$ does not satisfy $\Z$-Orderability. To see this, suppose for contradiction that it was, with $\lambda: \C \to \Z$ as the corresponding ordering. Then, $G := \lambda(c_0) - \lambda(c_F) < \infty$. However, since $\mu$ is a bijection, for any such $G$, there exist $G + 1$ calls $c_1, \dots, c_{G+1} \in \C \setminus \set{c_0, c_F}$ such that
    \[c_0 \succ_{x^*} c_1 \succ_{x^*} \cdots \succ_{x^*} c_{G+1} \succ_{x^*} c_F\]
    It follows that $\lambda(c_{G+1}) \le \lambda(c_0) - G$, but then we must have that $\lambda(c_{G+1}) \le \lambda(c_F)$ and we have arrived at a contradiction. 
\end{proof}

\subsection{Monotone Priorities}\label{sec:order-mps}
Although Existence does not imply $\Z_{\le \lambda_{\texttt{max}}}$-Orderability, in this subsection, we show how the five axioms Existence, Priority, Extension, Reducibility, and Independence of Irrelevant Calls do jointly imply this nice property. We do this concretely by showing that $\R^*(p)$, the unique rights system that satisfies all five axioms, is $\Z_{\le \lambda_{\texttt{max}}}$-Orderable. We show this constructively by formally associating $\R^*(p)$ with the Monotone Priority System as described in Section \ref{sec:intro-mps}.

Let $\lambda_{\texttt{max}} \in \Z$ be an integer denoting the maximum priority that a call can have. We define a Monotone Priority System $\lambda: \C \to \Z_{\le \lambda_{\texttt{max}}}$ to be a map that defines global integer priorities over all calls, such that for any call $c \in \C$, $\lambda(c) \le \min\set{\lambda(c') \mid c' \in \tr(c)}$. Notice that this definition coincides with the $\texttt{VALID}$ definition from Section \ref{sec:intro-mps}. We first show that any element of $\R^*(p)$ is orderable by a Monotone Priority System.
\begin{theorem}\label{thm:mpsexist}
    For any $p: \C \to \mb X$ with deployment order $>_D$ and $(\succ_x)_{x \in \mb X} \in \R^*(p)$, there exists a Monotone Priority System $\lambda: \C \to \Z_{\le \lambda_{\texttt{max}}}$ such that for any $x \in \mb X$ and $t, t' \in \C_x(p)$, $t \succ_x t' \implies \lambda(t) > \lambda(t')$.
\end{theorem}

\begin{proof}
    Let $(\succ_x)_{x \in \mb X} \in \R^*(p)$. We define $\lambda$ inductively using the deployment order $>_D$. Let $x_0, x_1, \dots \in \mb X$ be the unique ordering of all contracts in $\mb X$ such that $x_i >_D x_j$ if $i < j$. Noting that $\bigcup_{x \in \mb X} \C_x(p) = \C$, we constructively define $\lambda$ by defining its value on each $\C_x(p)$. 
    
    First, we let $\set{c_1^0, \dots, c_{N_0}^0} := \C_{x_0}(p)$ be the child calls of $x_0$ under $p$. Further, let $c_1^0, \dots, c_{N_0}^0$ be such that
        \[c_i^0 \succ_{x_0} c_j^0 \implies i < j\]
        That such an ordering exists follows directly from the fact that any strict partial order can be extended to a total order. We then define $\lambda(c_i^0) := \lambda_{\texttt{max}}-i$ for all $i \in [N_0]$. Notice that as $\succ_{x_0}$ must be admissible when restricted to $\C_{x_0}(p)$, it must be true that $\lambda(c) \le \min\set{\lambda(c') \mid c' \in \tr(c)}$ for all $c \in \C_{x_0}(p)$. Thus, in addition to defining $\lambda$ for all calls in $\C_{x_0}(p)$, we have also shown that $\lambda$ is a valid Monotone Priority System when restricted to these calls.
        
        Next, for any $k > 0$ such that $x_k$ is not terminal in the deployment ordering, let $\set{c_1^k, \dots, c_{N_k}^k} := \C_{x_k}(p)$. Again, further ensure that 
         \[c_i^k \succ_{x_k} c_j^k \implies i < j\]
         We now inductively define $\lambda(c_i)$. Starting with $c_1^k$, we define 
         \[\lambda(c_1^k) := \min\set{\lambda(c) \mid c \in \tr(c_1^k)},\]
         defining it as $\lambda_{\texttt{max}}- 1$ in the case that $\tr(c_1^k) = \emptyset$. This is well-defined as we have already inductively defined $\lambda$ for any such $c \in \tr(c_1^k)$. Next, for any $1 < i < N_k$, we define
         \[\lambda(c_i^k) := \min\pa{\lambda(c^k_{i-1}) - 1, \min\set{\lambda(c) \mid c \in \tr(c^k_i)}}\]
         Note that this is also well-defined as $\succ_x$ is admissible when restricted to $\C_{x_k}(p)$, and we have therefore inductively defined $\lambda$ in all cases where we reference it. Further, notice that it must further be true that $\lambda(c) \le \min\set{\lambda(c') \mid c' \in \tr(c)}$ for all $c \in \C_{x_k}(p)$. Thus, we have also shown that $\lambda$ is a valid Monotone Priority System when restricted to these calls.

         Finally, for the terminal contract $x_F$ (should it exist), we let $\set{c_1^F, c_2^F, \dots} := \C_{x_F}(p)$. Again, ensure that 
         \[c_i^F \succ_{x_F} c_j^F \implies i < j\]
         We then define $\lambda(c_i^F)$ inductively in the exact same way as before, noting that at each inductive step, the minima remain well defined as there are still only finite terms within it. An identical argument shows that $\lambda$ is a valid Monotone Priority System when restricted to $\C_{x_F}(p)$. It follows that $\lambda$ is indeed a Monotone Priority System.

         Next, we show that for this $\lambda$, if $t \succ_x t'$ for some $x \in \mb X$ and $t, t' \in \C$, then $\lambda(t) > \lambda(t')$. We prove this by splitting into two exhaustive cases.

         \begin{case}
             $p(t) = p(t') = x$
         \end{case}

         In this case, we have that $x = x_k$ for some $k \in \N \cup \set{F}$, and $t = c_i^k$ and $t' = c_j^k$ for $i < j$. Notice that in both the base and inductive case, we ensure that $\lambda(c_i^k) > \lambda(c_j^k)$ for $i < j$. Thus, $\lambda(t) > \lambda(t')$ as desired.

         \begin{case}
             $p(t) = x$ and $p(t') \ne x$, and there exists a $c \in \tr(t')$ with $p(c) = x$ such that $t \succ_x c$
         \end{case}

         In this case, we have that $x = x_k$ and $p(t') = x_{k'}$ for some $k,k' \in \N \cup \set{F}$, and $t = c_i^k$, $c = c_j^k$ with $i < j$, and $t' = c_\ell^{k'}$. Noting that we must have $k' > k$ by the property of the deployment order (with $F$ being notated here as greater than all natural numbers), it follows that in the inductive step, \[\lambda(t) = \lambda(c_i^k) > \lambda(c_j^k) \ge \lambda(c_\ell^{k'}) = \lambda(t'),\]
         whence the desired result follows.
\end{proof}

Next, we show that there are no ``extra'' Monotone Priority Systems, in that every Monotone Priority System in fact globally orders an element of $\R^*(p)$. 
\begin{theorem}\label{thm:mpsimp}
    For any $p: \C \to \mb X$ with deployment order $>_D$, and Monotone Priority System $\lambda: \C \to \Z_{\lambda_{\texttt{max}}}$, there exists a $(\succ_x)_{x \in \mb X} \in \R^*(p)$ such that for any $x \in \mb X$ and  $c, c' \in \C$, $c \succ_x c' \implies \lambda(c) > \lambda(c')$
\end{theorem}

\begin{proof}
    For any $x \in \mb X$ and $t, t' \in \C_x(p)$, notice first that there exists an admissible $\succ_x^*$ over $\C_x(p)$ such that $\lambda(t) > \lambda(t') \iff t \succ_x^* t'$. Let $(\succ_x)_{x \in \mb X}$ be the unique element of $\R^*(p)$ that is associated with $(\succ_x^*)_{x \in \mb X}$. Next, let $x \in \mb X$ be an arbitrary contract, and let $c, c' \in \C$ be any calls such that $c \succ_x c'$. We show that $\lambda(c) > \lambda(c')$. We prove this by splitting into two exhaustive cases.

         \begin{case}
             $p(c) = p(c') = x$
         \end{case}

         In this case, we have by construction that $c \succ_x c' \implies c \succ_x^* c' \implies \lambda(c) > \lambda(c')$.

         \begin{case}
             $p(c) = x$ and $p(c') \ne x$, and there exists a $f \in \tr(c')$ with $p(f) = x$ such that $c \succ_x f$
         \end{case}

         In this case, we have that $c \succ_x f \implies c \succ_x^* f \implies \lambda(c) > \lambda(f)$. Next, by the definition of a Monotone Priority System, it follows that $\lambda(c') \le \lambda(f)$, whence the desired result follows. This proves the first part of the Theorem.
\end{proof}

\section{Discussion}\label{sec:dis}

In this work, we defined the Monotone Priority System, and showed that it is the unique system that satisfies five natural axioms. The system allows contracts to non-trivially specify additional sequencing constraints as a part of block validity while remaining simple enough to plausibly be implemented on current infrastructure. We next give a more detailed discussion of our five axioms. 

Theorem \ref{thm:ind} implies that if any one of our five axioms are dropped, there are other sequencing systems besides the Monotone Priority System that satisfy the other four. Are any of these systems worth considering? And if so, which axiom(s) would any such systems not satisfy? 

\paragraph{Crucial Axioms: Existence, Priority, and Extension.} It is very difficult to imagine any reasonable system not satisfying all three of these axioms. Losing Existence would result in adding meaningful additional constraints on transaction \textit{inclusion}, as opposed to simply imposing constraints on sequencing. It is difficult to imagine that being tractable. A system without Priority would lack a baseline level of expressiveness. A system without Extension is simply not enforceable, as users can then easily get around sequencing constraints by making new contracts. 

With that said, perhaps a reasonable system may exist that instead has a modified Extension property. For example, one could consider allowing contracts to specify inadmissible partial orderings over their constituent calls, requiring instead that the Extension property only hold for calls that are in other contracts. The system does become more complicated in this case: suppose that a contract $x$ with two calls $a$ and $b$ where $b \in \tr(a)$ prioritizes $a \succ_x b$ (and is therefore inadmissible). Further, suppose that another call $c$ in contract $x'$ references $a$. There is a nuanced question of determining whether the system should enforce that $a \succ_x c$. This is still a reducible constraint, as $b \in \tr(c)$ (since it is in $\tr(a)$ and $a \in \tr(c)$) and $a \succ_x b$. However, modifying the Monotone Priority System to accommodate this creates complication, as now there are restrictions on the contract developer's selections of priorities that depends not only on calls that are directly referenced, but also those that are indirectly referenced. If instead, the system does not enforce $a \succ_x c$, then we can keep a simple priority ordering system like the Monotone Priority System, but the exact Extension guarantee that is made becomes more complicated, though potentially still practical. 

\paragraph{Negotiable Axioms: Reducibility, Independence of Irrelevant Calls.} Losing either of these axioms would result in systems where it is meaningfully more complicated to understand, from the perspective of a contract developer, how modifications to the contract's code will impact sequencing. However, these axioms seem potentially less immediately necessary for a practical system than the other three.\\

The existence of a static, global integer ordering (as discussed in Section \ref{sec:order}) is also a very convenient property to have to keep block building tractable. It allows the simple sequencing algorithm discussed in Section \ref{sec:intro-mps-sim} to find a valid block given a collection of transactions. While it's possible that there could be a reasonable rights system that lacks a global integer priority ordering, but still admits a practical sequencing algorithm algorithm, losing this property would make block building meaningfully less tractable.

A final hidden axiom throughout this paper is statelessness. It is difficult to imagine being able to express state-dependent sequencing constraints while keeping block building tractable, as the relative sequencing of transactions affects the underlying state. We leave open the question of whether another sequencing system that makes different, but practically relevant tradeoffs exists. 

\section{Acknowledgments}

I would like to thank Noam Nisan and Tim Roughgarden for helpful discussions regarding the description and framing of this system.

\bibliographystyle{alpha}
\bibliography{refs}

\end{document}

%% file: settings.tex
\usepackage{amsmath}
\usepackage{tcolorbox}
\usepackage{amssymb}
\usepackage{amsthm}
\usepackage{mathtools}
\usepackage{accents}
\usepackage{systeme}
\usepackage{tikz}
\usepackage{enumerate}
\usepackage{tikz-cd}
\usepackage{enumitem}
\usepackage{bbm}

\newtheorem{theorem}{Theorem}

\newtheorem{definition}{Definition}
\newtheorem{lemma}{Lemma}
\newtheorem{case}{Case}[subsection]
\newtheorem{axiom}{Axiom}
\newtheorem{assumption}{Assumption}

\newenvironment{amatrix}[1]{%
  \left[\begin{array}{@{}*{#1}{c}|c@{}}
}{%
  \end{array}\right]
}

\newcommand{\R}{\mathbb{R}}

\newcommand{\bra}[1]{\left[#1\right]}
\newcommand{\pa}[1]{\left(#1\right)}
\newcommand{\set}[1]{\left\{ #1 \right\}}
\newcommand{\abs}[1]{\left|#1\right|}
\newcommand{\Q}{\mathbb{Q}}
\newcommand{\Z}{\mathbb{Z}}
\newcommand{\N}{\mathbb{N}}
\newcommand{\C}{\mathbb{C}}

\newcommand{\comp}[1]{\overline{#1}}

\newcommand{\B}[1]{\boldsymbol{#1}}

\newcommand{\tr}{\operatorname{tr}}
\newcommand{\mb}{\mathbb}